\documentclass{article}

\usepackage{booktabs} 

\pdfoutput=1

\usepackage{amsmath, amsthm, amssymb}
\usepackage{algpseudocode,algorithm,algorithmicx}
\usepackage{mathtools}
\usepackage{comment} 
\usepackage{color-edits} 
\usepackage{tcolorbox}
\usepackage{xfrac}
\usepackage{hyperref}
\usepackage{setspace}

\allowdisplaybreaks

\definecolor{mygreen}{RGB}{20,120,60}

\hypersetup{
     colorlinks=true,
     citecolor= mygreen,
     linkcolor= mygreen
}

\usepackage[margin=1in]{geometry}

\newtheorem{theorem}{Theorem}
\newtheorem{claim}[theorem]{Claim}
\newtheorem{lemma}[theorem]{Lemma}

\newtheorem{definition}{Definition}

\newtheorem{observation}[theorem]{Observation}

\title{Almost Optimal Stochastic Weighted Matching with Few Queries\footnote{A preliminary version of this paper appeared at EC 2018.}}

\author{
\texorpdfstring{\hspace*{-8pt}}{}%
\begin{tabular}{c} Soheil Behnezhad\thanks{Portions of this work were completed while the author was an intern at Upwork.}\\ University of Maryland\\\texttt{soheil@cs.umd.edu}\end{tabular} \and
\begin{tabular}{c} Nima Reyhani\\Upwork\\\texttt{nimar@upwork.com}\end{tabular}}

\date{\vspace{0.5cm}Last update: May 2018\vspace{0.2cm}}

\begin{document}

\newcommand{\xhdr}[1]{\vspace{0.30cm}\noindent\textbf{#1}}

\addauthor{sb}{blue}    
\addauthor{nr}{red}  

\newcommand{\E}[0]{\ensuremath{\mathbb{E}}}

\newcommand{\matching}[1]{\ensuremath{M(#1)}}
\newcommand{\expmatching}[1]{\ensuremath{\mathbb{M}\lbrack #1 \rbrack}}

\newcommand{\diff}[2]{\ensuremath{#1\Delta#2}}

\newcommand{\opt}[0]{\ensuremath{\textsc{opt}}}

\newcommand{\tsum}[0]{\ensuremath{\textstyle\sum}}

\newcommand{\etal}[0]{\textit{et al.}}

\newcommand{\restate}[3]{\vspace{0.3cm}\noindent\textbf{#1~#2.} (Restated) #3\vspace{0.3cm}}
\DeclareRobustCommand{\mybox}[2][gray!20]{%
\begin{tcolorbox}[
        left=0pt,
        right=0pt,
        top=0pt,
        bottom=0pt,
        colback=#1,
        colframe=#1,
        width=\dimexpr\textwidth\relax, 
        enlarge left by=0mm,
        boxsep=5pt,
        arc=0pt,outer arc=0pt,
        ]
        #2
\end{tcolorbox}
}

\newcommand{\coloredalign}[1]{
	\vspace{4pt}
	\mybox{
		\vspace{-5pt}
			#1
		\vspace{-10pt}
	}
	\vspace{-4pt}
}

\maketitle

\begin{abstract}
We consider the \emph{stochastic matching} problem. An edge-weighted general (i.e., not necessarily bipartite) graph $G(V, E)$ is given in the input, where each edge in $E$ is \emph{realized} independently with probability $p$; the realization is initially unknown, however, we are able to \emph{query} the edges to determine whether they are realized. The goal is to query only a small number of edges to find a {\em realized matching} that is sufficiently close to the maximum matching among all realized edges. This problem has received a considerable attention during the past decade due to its numerous real-world applications in kidney-exchange, matchmaking services, online labor markets, and advertisements.

Our main result is an {\em adaptive} algorithm that for any arbitrarily small $\epsilon > 0$, finds a $(1-\epsilon)$-approximation in expectation, by querying only $O(1)$ edges per vertex. We further show that our approach leads to a $(1/2-\epsilon)$-approximate {\em non-adaptive} algorithm that also queries only $O(1)$ edges per vertex. Prior to our work, no nontrivial approximation was known for weighted graphs using a constant per-vertex budget. The state-of-the-art adaptive (resp. non-adaptive) algorithm of Maehara and Yamaguchi [SODA 2018] achieves a $(1-\epsilon)$-approximation (resp. $(1/2-\epsilon)$-approximation) by querying up to $O(w\log{n})$ edges per vertex where $w$ denotes the maximum integer edge-weight. Our result is a substantial improvement over this bound and has an appealing message: No matter what the structure of the input graph is, one can get arbitrarily close to the optimum solution by querying only a constant number of edges per vertex.

To obtain our results, we introduce novel properties of a generalization of \emph{augmenting paths} to weighted matchings that may be of independent interest.
\end{abstract}

\clearpage

\section{Introduction}

\setstretch{1.15}
We consider the \emph{stochastic weighted matching} problem where the goal is to find a maximum weighted matching under \emph{uncertainty} assumption about the input graph. More precisely, the input contains a general (i.e., not necessarily bipartite) weighted graph $G=(V, E)$ and a constant probability $p > 0$. We are unsure about the edges of $G$ that exist in the actual \emph{realization}, however, the assumption is that each edge of $G$ is \emph{realized} independently with probability $p$. An algorithm is able to \emph{query} an edge in $G$ to determine whether it exists in the actual realization. The goal is to find a sufficiently large weighted matching in the realized graph by querying only a small number of the edges. An {\em adaptive} algorithm in this setting can have multiple rounds of adaptivity where the queries conducted in each round of adaptivity may depend on the results of the previous rounds. A {\em non-adaptive} algorithm, on the other hand, has only one round of adaptivity. That is, a non-adaptive algorithm should pick a degree bounded subgraph $H$ of $G$ whose expected matching is sufficiently close to that of $G$.

Ignoring the constraint on the number of queries, a trivial solution is to query all the edges of $G$ and report the maximum weighted matching among the realized edges. However, for many applications of the stochastic matching problem (which we elaborate on in the next section), querying an edge is \emph{time consuming}. We are, therefore, interested in the tradeoff between the fraction of the omniscient optimum solution that we achieve and the number of edges that are queried.

Our main contribution is an adaptive algorithm, that for any arbitrarily small constant $\epsilon > 0$, finds in expectation, a $(1-\epsilon)$-approximation of the maximum weighted matching in the realized graph, by making only a {\em constant} (dependant only on $\epsilon$ and $p$) number of per-vertex queries, in a constant number of rounds. We also show that it is possible to obtain a $(1/2-\epsilon)$-approximation via a non-adaptive algorithm. This result, apart from its theoretical importance, has an appealing practical message: No matter what the structure of the input graph is, one can get arbitrarily close to the optimum solution by querying only a constant number of edges per vertex.\footnote{We note that our algorithms works even if the realization probability of the edges are different as long as they are all bounded away from 0 by a constant. In such cases, it suffices to define $p$ to be the minimum realization probability among all the edges.} 

It is worth mentioning that the dependence on both $\epsilon$ and $p$ is crucial for any algorithm that achieves a $(1-\epsilon)$-approximation. Consider the simple example of a star graph where one vertex $v$ is connected to all other vertices by edges of weight 1. The expected weight of the omniscient maximum matching is $1 - o(1)$ since the matching size would be 1 if at least one of the edges are realized. To ensure that with probability at least $1-\epsilon$ one of the queried edges is realized, one needs to query at least $\Omega(\log{(1/\epsilon)}/p)$ edges of $v$.\footnote{Note that although the matching size here is either 0 or 1, a $(1-\epsilon)$-approximation is well defined because we are comparing the {\em expected} size of the realized matching to that of the original graph. Another closely related example is a complete graph for which one can show existence of a perfect matching w.h.p. and follow the same approach to prove the same lower bound for {\em every} vertex (e.g., see \cite{DBLP:conf/sigecom/AssadiKL16}).}


Stochastic matching settings have been studied extensively in the past decade after the initial paper of ~\cite{chen2009approximating}.
Motivated mainly by its application in \emph{kidney exchange}, this natural variant of stochastic matching problem was initially introduced in \cite{DBLP:conf/sigecom/BlumDHPSS15} and has received significant attention ever since \cite{DBLP:conf/sigecom/BlumDHPSS15, DBLP:conf/sigecom/AssadiKL16, DBLP:conf/sigecom/AssadiKL17, maeharatakanori}. In the literature, most algorithms work only for unweighted graphs~\cite{DBLP:conf/sigecom/BlumDHPSS15, DBLP:conf/sigecom/AssadiKL16, DBLP:conf/sigecom/AssadiKL17} where the goal is to approximate maximum cardinality matching. The only exception is the very recent paper of Maehara and Yamaguchi~\cite{maeharatakanori} that achieves a $(1-\epsilon)$-approximation via an adaptive algorithm for general (resp. bipartite) weighted graphs with $O(w \log{n}/\epsilon p)$ (resp. $O(w/\epsilon p)$) queries per vertex where $w$ denotes the maximum edge weight assuming that all edge weights are positive integers. They also show that a $(1/2-\epsilon)$-approximation is achievable via a non-adaptive algorithm using the same number of queries per vertex. Our result is a substantial improvement over this bound as it has no dependence on the structural properties of the input graph such as $n$ or $w$.

Our techniques are based on a novel set of definitions and lemmas for weighted \emph{augmenting components} that generalize augmenting paths to weighted graphs. Despite being useful objects when the goal is to find a large cardinality matching, augmenting paths have rarely been useful for weighted graphs in different settings of the matching problem. These properties may be of independent interest for generalizing augmenting path based techniques to weighted graphs.\footnote{There is also a notion of ``augmenting paths" for flow algorithms. Here we only refer to augmenting path based techniques that are used in matching algorithms.}

\subsection{Applications}\label{sec:applications}
\paragraph{Kidney exchange.} One of the main applications of the stochastic matching problem is in \emph{kidney exchange}, which includes two types of participants: an organ donor and an organ receiver. The donor is a healthy living individual willing to donate one of his/her two kidneys to the organ receiver (patient), without much of life threatening side effects. Usually, the donor is among family and friends. However, the kidney of a donor might not be a compatible match for the patient due to physiological reasons including the incompatible blood type, tissue-type, etc. The goal in kidney exchange is to provide a platform to swap kidneys between pairs of organ donors and patients in order to find a match. 

Usually the donors' and the patients' medical records (e.g. blood type) are available as an early indicator of the compatibility. The basic information is not conclusive, therefore, extra medical laboratory tests such as antibody screening is required to better estimate the odds of a successful transplant. The extra medical tests are time consuming and costly; specially for the patients who have been long waiting for the transplant or are in critical life condition.

Kidney exchange can be seen as a matching problem where each incompatible donor-patient pair is a vertex and there exists an edge between two vertices when the two kidneys can be swapped between pairs. A maximum matching of the donor-patients graph finds the maximum number of the donor-patient pairs who can swap kidneys. For many reasons, such as geographic proximity, wait time, age, etc. \cite{dickerson2012optimizing}, the design committee may introduce weights on the edges. This means, depending on the matching policy and the edge weights, maximizing the matching cardinality is not necessarily the best solution.

The stochastic matching problem, here, helps in reducing the number of costly and time consuming extra medical tests which should be performed to find the compatibility of the donor-patient pairs. The algorithm selects (i.e., queries) only a small set of donor-patient pairs for extra medical lab-tests so that in expectation, a sufficiently large matching exists within the selected pairs that pass the test in the realized graph.

There has been a plethora of studies on kidney exchange, particularly on stochastic settings \cite{akbarpour2014dynamic, anderson2015dynamic, anderson2015finding, awasthi2009online, dickerson2012dynamic, dickerson2013failure, dickerson2015futurematch, manlove2012paired, unver2010dynamic}. We refer interested readers to the paper of \cite{DBLP:conf/sigecom/BlumDHPSS15} for a more detailed discussion.

\paragraph{Online labour markets.} Traditional full-time employment has continuously given way to {\em flexible} contract work. This has resulted in a growth of demand for the services of the companies such as Upwork, Guru, Freelancer and Fiverr that facilitate working relationships between freelancers and employers. The parties in such markets have often more options than they can consider and interview with. In the stochastic matching context, job descriptions and the freelancers can be represented by vertices in a bipartite graph where an edge between a freelancer and a job posting would represent whether the freelancer is a good fit for the job. While initial job descriptions by an employer rule out some of the edges, there is no way of having a full knowledge about the actual graph beforehand. Here, querying an edge maps to the pair having an interview, and therefore, the stochastic matching procedure helps in minimizing the number of interviews while achieving a near optimal solution. Edge weights, in this context, help in maximizing, say, the social welfare or other objective functions.

By analogy, stochastic matching also contributes to other \emph{matching services}, such as online dating services.

\subsection{The Model}\label{sec:model}
In the \emph{stochastic setting}, for input graph $G=(V, E)$, the edges in $E$, independently from each other, are realized w.p. $p$.\footnote{Throughout the paper, we use \textit{w.p.} to abbreviate ``with probability".} For a subset $E'$ of $E$, we define $E'_p$ to be a random variable corresponding to different \emph{realizations} of $E'$ where each edge in $E'$ is \emph{realized} independently w.p. $p$.

In the \textit{stochastic matching problem}, a graph $G = (V, E)$ with a non-negative weight $w_e$ for each edge $e \in E$ as well as the realization probability $p$ are given the input. A realization $\mathcal{E}$ of $E$ that is drawn from $E_p$ is initially fixed but is unknown to us.  Our goal is to find a weighted matching of $\mathcal{E}$ that is close to \matching{\mathcal{E}}, where \matching{.} denotes the maximum weighted matching of its input edge set.\footnote{We may slightly abuse the notation and use \matching{.} to both refer to a maximum weighted matching and its weight.} An algorithm is allowed to \emph{query} the edges in $E$ to determine whether they are in $\mathcal{E}$ or not. 

An {\em adaptive} algorithm proceeds in separate rounds. At each round $i$, based on the results of the queries thus far, the algorithm queries a new subset of the edges $Q_i$ in parallel. The goal is to maximize $\E[\matching{\cup_{i=1}^{R}Q'_i}]$, where $Q'_i$ denotes the edges in $Q_i$ that exist in $\mathcal{E}$, and guarantee that it is close to $\opt := \E[\matching{\mathcal{E}}]$. These expectations are taken over the randomness of the realization $\mathcal{E}$. The total number of rounds that an adaptive algorithm takes is called its \emph{degree of adaptivity}. A {\em non-adaptive} algorithm on the other hand has only one round of adaptivity.

We hope to have algorithms in which the number of per-vertex queries and the degree of adaptivity are both independent of the structure of the input graph such as the edge weights or  $n$.

\subsection{Related Work}
Directly related to the model that we consider in this work are \cite{DBLP:conf/sigecom/BlumDHPSS15, DBLP:conf/sigecom/AssadiKL16, DBLP:conf/sigecom/AssadiKL17, maeharatakanori}, with the only difference that (except for \cite{maeharatakanori}) their results only hold for unweighted graphs. More specifically, Blum \textit{et al.}~\cite{DBLP:conf/sigecom/BlumDHPSS15} give an adaptive algorithm that achieves a $(1-\epsilon)$-approximation by querying $\frac{\log(2/\epsilon)}{p^{2/\epsilon}}$ per-vertex edges, in the same number of rounds. Assadi \textit{et al.}~\cite{DBLP:conf/sigecom/AssadiKL16} improve this result to get rid of the exponential dependence on $1/\epsilon$. In particular, they show that it is possible to achieve the same approximation factor by querying only $O(\log(1/\epsilon p)/\epsilon p)$ edges, with a very similar algorithm. Very recently, Maehara and Yamaguchi~\cite{maeharatakanori}, among some other related problems, considered the stochastic weighted matching problem. They provided a $(1-\epsilon)$-approximation via an adaptive algorithm for general (resp. bipartite) weighted graphs with $O(w \cdot \log{n}/\epsilon p)$ (resp. $O(w/\epsilon p)$) queries per vertex where $w$ is the maximum edge weight assuming that all edge weights are positive integers. See Section~\ref{subsec:comparison} for a comparison of our techniques to the techniques used in these papers.

For non-adaptive algorithms, all of the papers above also obtain a $(1/2-\epsilon)$-approximation via a non-adaptive algorithm. The state-of-the-art non-adaptive algorithm for unweighted graphs is the algorithm of Assadi \textit{et al.}~\cite{DBLP:conf/sigecom/AssadiKL17}, which achieves a slightly better than half approximation for unweighted graphs. 

Stochastic matching has received a lot of attention in the last decade due to its numerous applications. A very well-studied setting is the \emph{query-commit} model that was first introduced by Chen \textit{et al.}~\cite{chen2009approximating}. In the query-commit model, a queried realized edge has to be included in the final matching \cite{adamczyk2011improved, bansal2012lp, chen2009approximating, costello2012stochastic, gupta2013stochastic}. Another related setting is that of \cite{blum2013harnessing}. In this setting, the algorithm is allowed to query at most two incident edges of each vertex and the goal is to find the optimal subset of edges to query.

\section{Our Results and Techniques}
Our main result is the following theorem which is formally given in Section~\ref{sec:adaptive} as Theorem~\ref{thm:adaptive}.

\begin{theorem}[Informal]\label{thm:adaptiveinformal}
	For any given $\epsilon > 0$, there exists a poly-time adaptive algorithm to achieve a $(1-\epsilon)$-approximation of the stochastic weighted matching problem by querying only $O(\frac{1}{\epsilon p^{4/\epsilon}})$ edges per vertex in $O(\frac{1}{\epsilon p^{4/\epsilon}})$ rounds of adaptivity.
\end{theorem}

We further show that one can obtain a $(0.5-\epsilon)$-approximation via a non-adaptive algorithm with the same number of per-vertex queries.

\begin{theorem}[Informal]\label{thm:nonadaptiveinformal}
	For any given $\epsilon > 0$, there exists a poly-time non-adaptive algorithm to achieve a $(0.5-\epsilon)$-approximation of the stochastic weighted matching problem by querying only $O(\frac{1}{\epsilon p^{4/\epsilon}})$ edges per vertex.
\end{theorem}

Here we focus on the adaptive algorithm and explain the main technical ingredients used in proving Theorem~\ref{thm:adaptiveinformal}. In brief, our algorithm is as follows: take the maximum weighted matching in $G$, query its edges, and remove from $G$ the queried edges that are known to be unrealized. Repeat this for $R$ (a parameter to be determined later) rounds; then report the maximum weighted matching among the realized queried edges.

Inspired by the analysis of \cite{DBLP:conf/sigecom/BlumDHPSS15} for the unweighted variant of the problem, we first provide a proof sketch to show that this algorithm achieves a $(1-\epsilon)$-approximation for $R = O(1)$, if all the edge weights are the same (which is equivalent to the case of unweighted graphs). We then focus on challenges in generalizing this approach to the general weighted case and describe the intuitions on how we overcome them.

The idea is to maintain a {\em realized matching} and iteratively augment it at each round until its size becomes as large as $(1-\epsilon)\opt$. Let us denote by $O_r$ the maximum realized matching among the queried edges by round $r$. Note that if $|O_{r-1}| < (1-\epsilon)\opt$, the next matching that we pick (which is definitely of size at least $\opt$) augments $O_{r-1}$ by \emph{many} vertex disjoint augmenting paths of size at most $O(1/\epsilon)$. Any of these augmenting paths increases the matching size only if all of its edges are realized. However, since the length of each of the augmenting paths is at most $O(1/\epsilon)$, the probability that all of the edges of any of the augmenting path is realized is at least $p^{O(1/\epsilon)}$. Hence roughly after $1/p^{O(1/\epsilon)}$ rounds, we achieve a $(1-\epsilon)$-approximation in expectation.

This proof contains two canonical parts: (i) there are \emph{many} vertex disjoint augmenting paths; (ii) each of them has a relatively \emph{small length}. Roughly speaking, these claims are correct since in unweighted graphs, each augmenting path increases the size of a matching by at most one edge. In fact, for weighted graphs, one may give an example that refutes both of these assumptions at the same time.

This difference, is perhaps, one of the main reasons that in the literature of the matching problem, augmenting paths have extensively been considered for the unweighted graphs, where in contrast, they have rarely been useful when the graph under consideration is weighted. In this work, we propose a novel direction to generalize augmenting paths to weighted graphs that may be of independent interest.

Consider two weighted matchings $M_L$ and $M_H$ of the same graph with $w(M_L) < (1-\epsilon)w(M_H)$. Further consider the graph $M_H \Delta M_L := (M_H \cup M_L) - (M_H \cap M_L)$ which is a set of connected components that are either paths or cycles since the degree of each vertex in it is at most 2. We call each of these connected components an {\em alternating component}. Moreover, we call each alternating component an {\em augmenting component} if the total weight of edges of $M_H$ in it is more than that of $M_L$.

For any alternating component $C$ in $M_H \Delta M_L$, we denote by $\Delta_C$ the amount of weight that is added to the base matching $M_L$ after we swap the edges of $M_L$ in $C$ with that of $M_H$ (we call this process {\em augmenting} $C$). Thus, we have $\sum_{C \in M_H \Delta M_L}\Delta_C = w(M_H) - w(M_L)$. We refer to $\Delta_C$ as the value of alternating component $C$. Hence, an augmenting component is an alternating component with a positive value. We further define the \emph{augmenting edges} of $C$, denoted by $Q_C$, to be the set of edges in $C$ that are in $M_H$. Moreover, we define the \emph{normalized length} of $C$, denoted by $L_C$, to be $w(Q_C)/\Delta_C$ where $w(Q_C)$ denotes the total weight of the edges in $Q_C$. 

In the stochastic matching context, $M_L$ will correspond to our maintained realized matching and $M_H$ will correspond to the matching that we pick in the next iteration of the algorithm. The main theorem of this paper is obtained by coupling the following two lemmas:

\begin{enumerate}
	\item Querying the edges of an augmenting component $C$, increases the expected weight of the maximum realized matching by at least $\Omega(p^{O(L_C)}\Delta_C)$ in expectation. (See Lemma~\ref{lem:augpathaddition} for the formal statement.)
	\item There is always a set $S$ of \emph{high impact} (i.e., $\sum_{C \in S}\Delta_C$ is sufficiently large) vertex-disjoint augmenting components of {\em small} normalized length. (See Lemma~\ref{lem:shortpaths} for the formal statement.)
\end{enumerate}

\paragraph{Intuition behind (1).} Fix an augmenting component $C$. Denote the actual length (i.e., the number of edges) of $C$ by $k$. If $k \leq O(L_C)$ the proof of the first lemma is trivial. Indeed, all the edges of $C$ are realized with probability $p^k$ and if they are all realized, $C$, after augmentation, adds $\Delta_C$ to the weight of the realized matching that we maintain. Hence in expectation $p^k\Delta_C$ ($\geq p^{O(L_C)}\Delta_C$) is added to the weight of our realized matching as desired. Roughly speaking, for the case when $k$ is much larger than $L_C$, we cannot \emph{afford} the probability by which the whole augmenting component is realized. Instead, we decompose the component into a set of vertex disjoint {\em sub-components} by removing some of its non-queried edges. This results in having smaller components that have a higher chance of being completely realized. We then argue that the amount of weight that these smaller components add to the matching is larger than $\Omega(p^{O(L_C)}\Delta_C)$ even considering the weight of the removed edges.

\paragraph{Intuition behind (2).} The proof of this lemma is mainly based on the assumption that $w(M_L) < (1-\epsilon)w(M_H)$. If the normalized length $L_C = w(Q_C)/\Delta_C$ of an augmenting component $C$ is large, then its augmenting edges $Q_C$ should have a large total weight while its value $\Delta_C$ should be relatively small. This implies, intuitively, that a large portion of the edges in the heavier matching $M_H$ are {\em wasted} without augmenting the base matching $M_L$ by a large enough weight. Nonetheless, since we have $w(M_L) < (1-\epsilon)w(M_H)$, we can argue that this cannot happen for an arbitrarily large portion (based on their weights) of augmenting components.

\subsection{Comparison of Techniques with Prior Works}\label{subsec:comparison}
The most relevant paper to our work, in terms of the techniques that are used, is that of \cite{DBLP:conf/sigecom/BlumDHPSS15} which only considers the unweighted version of the problem and the challenges in generalizing their analysis to the weighted case were discussed above. In summary, to achieve our results, we focus on properties of weighted augmenting paths. An important step is to relax the requirement that every edge in an augmenting path should be realized for it to be useful. We then show that it is possible to get $\epsilon$-close to the optimum weighted matching by constant per-vertex queries.

Assadi \textit{et al.}~\cite{DBLP:conf/sigecom/AssadiKL16}, improve the result of \cite{DBLP:conf/sigecom/BlumDHPSS15} by achieving the same approximation factor while querying only $O(\log(1/\epsilon p)/\epsilon p)$ edges per vertex. Their ideas are mainly based on the intuition that even if parts of augmenting paths are realized, one can ``patch" them together to create new augmenting paths and use Tutte-Berg to formalize this. Unfortunately, no equivalent generalization of Tutte-Berg formula is known for weighted graphs and thus our techniques are inherently different. However, our analysis is also based on the fact that parts of the edges in an augmenting path may be enough to obtain a better weighted matching.

Maehara and Yamaguchi~\cite{maeharatakanori} gave the only known approximation on weighted graphs prior to our work. They show that it is possible to get a $(1-\epsilon)$-approximation for general (resp. bipartite) weighted graphs with $O(w \log{n})$ (resp. $O(w)$) queries per vertex where $w$ denotes the maximum edge weight assuming that all edge weights are positive integers. Their results are based on a general integer programming framework. More specifically, they formulate the stochastic matching problem by an integer program and use LP relaxations of maximum weighted matching on bipartite and general graphs to obtain their approximations. The dependence of the number of per-vertex queries on the edge weights' range ($w$) is intrinsically part of their framework, hence both of their approximations for bipartite and general graphs depend on $w$.

As an interesting side note, apart from minor differences, the descriptions of the algorithms proposed for the stochastic matching problem in the literature are all simple and similar, and the differences only appear in the analysis \cite{DBLP:conf/sigecom/BlumDHPSS15, DBLP:conf/sigecom/AssadiKL16, DBLP:conf/sigecom/AssadiKL17, maeharatakanori}.

\section{$(1-\epsilon)$-Approximation via an Adaptive Algorithm}\label{sec:adaptive}
\newcommand{\M}[1]{\ensuremath{M_{#1}}}
\newcommand{\Mt}[1]{\ensuremath{M^T_{#1}}}
\newcommand{\Mf}[1]{\ensuremath{M^F_{#1}}}
\newcommand{\Estar}[0]{\ensuremath{E^*}}

\newcommand{\baseweight}[1]{\ensuremath{b_{#1}}}
\newcommand{\length}[1]{\ensuremath{\ell_{#1}}}

In this section we give an adaptive algorithm that returns a matching of the realized graph whose expected weight is a $(1-\epsilon)$-approximation of the optimum solution for any arbitrarily small constant $\epsilon$. The number of per-vertex queries, as well as the degree of adaptivity of our algorithm is $O(1)$ (it only depends on $\epsilon$ and $p$ --- and not properties of the input graph such as $n$ or the edge weights).

The algorithm, formally given as Algorithm~\ref{alg:adaptive}, is simple: In each round, among the edges that are either not queried yet or are queried and realized, it finds a maximum weighted matching and queries its edges. Finally after a sufficiently large (but constant) number of rounds, it reports the maximum weighted matching among the realized edges that had been queried.

\begin{algorithm}
  \caption{Adaptive algorithm for weighted stochastic matching: $(1-\epsilon)$-approximation.}
  \label{alg:adaptive}
  \begin{algorithmic}[1]
  	\Statex \textbf{Input:} Input graph $G=(V, E)$.
  	\Statex \textbf{Parameter:} $R$.
  	\State $\Estar{} \gets E$
	\For{$r = 1, \ldots, R$}
		\State Find a maximum weighted matching \M{r} in $(V, \Estar{})$.
		\State Query the edges in \M{r} and remove its non-realized edges from \Estar{}.
	\EndFor
	\State \Return{a maximum weighted matching among realized edges in $\cup_{r=1}^{R}\M{r}$.}
  \end{algorithmic}
\end{algorithm}

Note that $R$ is an upper bound on the per vertex queries and the degree of adaptivity of Algorithm~\ref{alg:adaptive}. The main focus of this section is to show that for a constant $R$, the weight of the matching that Algorithm~\ref{alg:adaptive} returns is in expectation a $(1-\epsilon)$-approximation of the omniscient optimum. Theorem~\ref{thm:adaptive}, which is the formal restatement of Theorem~\ref{thm:adaptiveinformal}, captures this. 

\begin{theorem}\label{thm:adaptive}
	For any graph $G=(V, E)$, and any arbitrarily small constant $\epsilon > 0$, Algorithm~\ref{alg:adaptive} returns a matching whose expected weight is at least $(1-\epsilon)\opt$ for $R = O(\frac{1}{\epsilon p^{4/\epsilon}})$.
\end{theorem}

We start by a set of definitions for weighted matchings in Section~\ref{sec:weighteddefs} and then proceed to prove Theorem~\ref{thm:adaptive} in Section~\ref{sec:adaptiveanalysis}

\subsection{Tools for Analyzing Weighted Matchings}\label{sec:weighteddefs}
Consider two weighted matchings $M_L$ and $M_H$ and assume that the total weight of matching $M_H$ is larger than that of $M_L$, i.e., $w(M_L) < w(M_H)$ where for any $E' \subseteq E$, we use $w(E') := \sum_{e\in E'}w_e$ to denote the total weight of edges in $E'$. Our goal in this section is to define a set of tools, among which, the most important is the notion of {\em augmenting components} which generalizes augmenting paths to weighted graphs.

Let us denote by $\diff{A}{B} := (A \cup B) - (A \cap B)$ the symmetric difference of two sets $A$ and $B$.
	Consider the graph $D := M_L\Delta M_H$ of edges that appear in exactly one of $M_L$ and $M_H$. Since the degree of each vertex in $D$ is at most 2 (due to the fact that both $M_1$ and $M_2$ are matchings), $D$ is a collection of vertex disjoint paths and cycles. We start by defining alternating components.
	
\begin{definition}[Alternating components]
	We call any connected component of $M_L \Delta M_H$ an {\em alternating component}. For an alternating component $C$, we define $B_C := C \cap M_L$ to be its edges in $M_L$ and define $Q_C := C \cap M_H$ to be its edges in $M_H$. We refer to $B_C$ as the {\em base} edges of component $C$ and refer to $Q_C$ as the {\em augmenting} edges of $C$.
\end{definition}

Note that alternating components in weighted graphs can be either cycles or paths, hence we refer to them as {\em components} instead of paths which is more common for cardinality matching in unweighted graphs. Next, we define a set of properties of these alternating components (see Figure~\ref{fig:augpath}).

\begin{definition}\label{def:augpathprops}
For any alternating component $C$:
\begin{itemize}
	\item We define the {\em value} $\Delta_C$ of $C$ to be $w(Q_C) - w(B_C)$.
	\item We define the {\em normalized length} $L_C$ of $C$ to be $w(Q_C)/\Delta_C$.
\end{itemize}
\end{definition}

Observe that by definition, the value of an alternating component may even be negative. We use the term {\em augmenting components} to refer to any alternating component $C$ with $\Delta_C > 0$. Furthermore, we also use the term {\em augmenting} a component $C$ to denote the process of updating the matching $M_L$ to be $M_L \backslash B_C \cup Q_C$ which clearly is a valid matching since all components are vertex disjoint. 

\begin{figure}[hbt]
  \centering
  \includegraphics[scale=0.7]{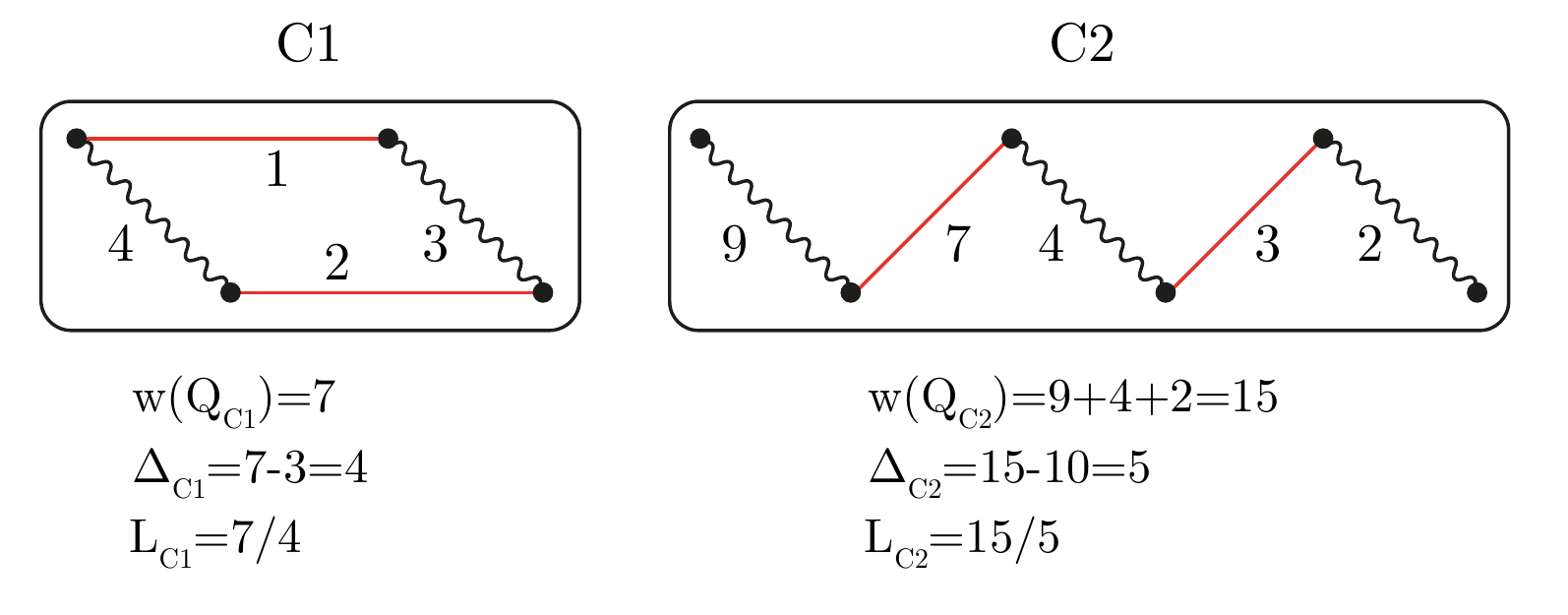}
  \caption{Two augmenting components and their properties. The wavy (black) edges are the edges in $M_H$. The solid (red) edges are the edges in $M_L$.}
  \label{fig:augpath}
\end{figure}

\subsection{Analysis}\label{sec:adaptiveanalysis}
Let \Mt{r} (resp. \Mf{r}) denote the edges in \M{r} that are realized (resp., are not realized). Furthermore, denote  by $O_r := \matching{\cup_{i=1}^{r}\Mt{r}}$ the maximum realized matching found by round $r$ (i.e., $O_R$ is what Algorithm~\ref{alg:adaptive} outputs). Let us denote by $\opt$ the weight of the omniscient optimum maximum weight matching. Our goal is to show that for any round $r-1$ with $w(O_{r-1}) < (1-\epsilon)\opt$, the matching $\M{r}$ picked in the next round of Algorithm~\ref{alg:adaptive} augments $O_{r-1}$ by a constant fraction of $\opt$ and use this to show that it takes only a constant number of rounds to obtain a $(1-\epsilon)$-approximation. We start by the following standard observation.

\begin{observation}\label{obs:standardforadaptive}
	For any $r \geq 1$, we have $w(\M{r}) \geq \opt$.
\end{observation}
\begin{proof}
	Matching $\M{r}$ is the maximum weighted matching over a set of edges that is a super set of the edges in the realization, hence $w(\M{r})$ is larger than the weight of $\opt$ which is the maximum weighted matching over the realized edges. 
\end{proof}

Consider the two matchings \M{r} and $O_{r-1}$ and define $U_r$ to be the set of augmenting\footnote{It is important that $U_r$ is the set of {\em augmenting} components in $\M{r}\Delta O_{r-1}$, and not all of its alternating components.} components in $\M{r}\Delta O_{r-1}$. We couple the following two lemmas to prove Theorem~\ref{thm:adaptive}.

\begin{lemma}\label{lem:augpathaddition}
Querying the edges of each augmenting component $C$ in $U_r$ increases the expected matching size by at least $p^{2 L_C}\Delta_C/2$.	
\end{lemma}

\begin{lemma}\label{lem:shortpaths}
If $w(O_{r-1}) < (1-\epsilon)w(M_r)$, then $\sum_{C \in U_r: L_C < 2/\epsilon} \Delta_C \geq \frac{\epsilon}{2} w(\M{r})$.	
\end{lemma}

We first show how with these two lemmas we can prove Theorem~\ref{thm:adaptive}.

\begin{proof}[Proof of Theorem~\ref{thm:adaptive}]
	By Observation~\ref{obs:standardforadaptive}, we know that for any round $r$, we have $w(M_r) \geq \opt$. Suppose we have not obtained a $(1-\epsilon)$-approximation by round $r-1$, then we have $$w(O_{r-1}) \leq (1-\epsilon) \opt \leq (1-\epsilon) w(M_r).$$
	Let $S$ be the set of augmenting components in $U_r$ with normalized length at most $2/\epsilon$. The inequality above satisfies the condition of Lemma~\ref{lem:shortpaths}, thus, we have 
	\begin{equation}\label{eq:deltaopt}
	\sum_{C \in S} \Delta_C \geq \frac{\epsilon}{2} w(M_r) \geq \frac{\epsilon}{2}\opt.	
	\end{equation}
	By Lemma~\ref{lem:augpathaddition}, querying the edges of each of the augmenting components in $U_r$ increases the expected matching size by at least $p^{2L_C}\Delta_C/2$. Therefore, 
	\begin{flalign*}
		&& \E[ w(O_r) - w(O_{r-1}) ] &\geq \frac{1}{2} \sum_{C \in U_r}p^{2L_C}\Delta_C &&\\
		&& &\geq \frac{1}{2} \sum_{C \in S}p^{2L_C}\Delta_C && \text{Since for any $C \in U_r$, $\Delta_C > 0$ and $S \subseteq U_r$.} \\
		&& &\geq \frac{1}{2} \sum_{C \in S}p^{4/\epsilon}\Delta_C && \text{Since for any $C \in S$, $L_C \leq 2/\epsilon$.}\\
		&& &\geq \frac{p^{4/\epsilon}}{2} \sum_{C \in S}\Delta_C\\
		&& &\geq \frac{\epsilon p^{4/\epsilon}}{4} \opt && \text{By (\ref{eq:deltaopt}).}\\
	\end{flalign*}
	This means that until we reach a $(1-\epsilon)$-approximation, we add, in expectation, a value of at least $\Omega(\epsilon p^{4/\epsilon} \opt)$ to our maximum weighted realized matching in each round. Therefore, in expectation, it takes at most $O(\frac{1}{\epsilon p^{4/\epsilon}})$ rounds to obtain a $(1-\epsilon)$-approximation.
\end{proof}

\subsubsection{Proof of Lemma~\ref{lem:augpathaddition}}

\begin{proof}[Proof of Lemma~\ref{lem:augpathaddition}]
	
Observe that the base edges of $C$ which are its edges that are in $O_{r-1}$ are already known to be realized by definition of $O_{r-1}$. However, the augmenting edges of $C$ which are the edges that are from matching $M_{r}$ might not have been queried before and thus might appear to be unrealized after being queried. The challenge, thus, is to show that in expectation the portion of augmenting edges of $C$ that will be realized augment $O_{r-1}$ by a total weight of at least $p^{2L_C}\Delta_C/2$.

Denote by $k$ the number of edges of $C$ that are from $M_r$. Since each of these edges are realized with probability at least $p$ independently from other edges, with probability at least $p^k$ all of them are realized and if they are all realized they augment the base matching $O_{r-1}$ by at least $\Delta_C$. Hence they add, in expectation, a total weight of at least $p^k\Delta_C$. If $p^k \Delta_C \geq p^{2L_C}\Delta_C/2$ we are done. However this might not be the case when $L_C$ is much smaller than $k$. To handle this, we show that even if  small portions of $C$ are realized, they augment the base matching by a large enough value and argue that the expected augmentation value is at least $p^{2L_C}\Delta_C/2$.

\begin{definition}\label{def:ordering}
Depending on the structure of component $C$, label its edges in the following way:
	\begin{enumerate}
		\item If $C$ is a path starting with an edge of $M_r$, denote its edges by $(q_1, b_1, q_2, b_2, \ldots)$ that is in the order that the edges appear in the path.
		\item if $C$ is a path starting with an edge of $O_{r-1}$, denote its edges by $(b_0, q_1, b_1, q_2, b_2, \ldots)$ that is in the order that the edges appear in the path.
		\item if $C$ is a cycle, denote its edges by $(q_1, b_1, q_2, b_2, \ldots)$ that is in the order that the edges appear in the cycle, starting from an arbitrary augmenting edge of $C$.
	\end{enumerate}
\end{definition}

	Suppose for ease of exposition that $\alpha := 2L_C$ is an integer. For any integer $i \in [\alpha]$, define $D_i$ to be the following set of edges of $M_r$ in $C$:
	\begin{equation*}
		D_i := \{ q_i, q_{i+ \alpha}, q_{i+2\alpha}, q_{i+3\alpha}, \ldots\}.
	\end{equation*}
	
	\begin{claim}\label{claim:shiftingsubsets}
		There exists some integer $i \in [\alpha]$ for which $w(D_i) \leq w(Q_C)/\alpha$.
	\end{claim}
	\begin{proof}
		It is easy to confirm by definition that for any $i, j \in [\alpha]$ with $i \not= j$, we have $D_i \cap D_j = \emptyset$ and that $\cup_{i\in[\alpha]}D_i = w(\cup_i \{q_i\}) =  w(Q_C)$. That is, the edges of $Q_C$ are partitioned by $D_i$'s into $\alpha$ disjoint subsets. Thus, the average total weight of $D_i$ for $i \in [\alpha]$ is $w(Q_C)/\alpha$, meaning that there should be some $i \in [\alpha]$ with $w(D_i)$ no more than this average value.
	\end{proof}
	
	Now take an arbitrary $i \in [\alpha]$ for which $w(D_i) \leq w(Q_C)/\alpha$ and ``delete" all the edges of $D_i$ from $C$. This procedure decomposes our augmenting component $C$ into a set $F=\{C_1, C_2, \ldots, C_t\}$ of smaller vertex disjoint connected components that we call {\em sub-components}. While each of these sub-components has the good property that it has a higher chance of being completely realized compared to the much longer original component $C$, we also need to take into account the loss of value resulted by deleting the edges of $D_i$ from \M{r}. Nonetheless, we argue that if we sum up the expected augmentation value of each of these vertex disjoint sub-components, the total expected value is as large as our desired value of $p^{2L_C}\Delta_C/2$.
	
	One can confirm by the labeling procedure defined in Definition~\ref{def:ordering} that after removal of the edges in $D_i$, each sub-component $C_j \in F$ has at most $\alpha$ edges of $Q_C$. Therefore, for any sub-component $C_j$, all the edges in $C_j \cap M_r$ will be realized with probability at least $p^\alpha = p^{2L_C}$.
	
	Denote by $\Delta_{C_j} := (C_j \cap M_r) - (C_j \cap O_{r-1})$ the {\em value} of sub-component $C_j$. By observation above, each sub-component $C_j$ augments its base matching by a total weight of at least $p^{2L_C}\Delta_{C_j}$ in expectation. Hence the total augmentation value of all sub-components is, in expectation, at least
	\begin{equation}\label{eq:augmentvaluesubcomponents}
		\sum_{C_j \in F} p^{2L_C} \Delta_{C_j} = p^{2L_C}\sum_{C_j \in F} \Delta_{C_j}.
	\end{equation}
	It only suffices to show that $\sum_{C_j \in F} \Delta_{C_j}$ is sufficiently large for our purpose. To see this, observe that $\sum_{C_j \in F} \Delta_{C_j}$ is essentially equal to the total value $\Delta_C$ of augmenting component $C$ minus the weight of the edges of $D_i$ since we removed them. However, since $w(D_i) \leq w(Q_C)/2L_C$ by Claim~\ref{claim:shiftingsubsets}, we have
	\begin{flalign}
		\nonumber && \sum_{C_j \in F} \Delta_{C_j} &\geq \Delta_C - \frac{w(Q_C)}{2L_C} &&\\
		\nonumber && &\geq \Delta_C - \frac{w(Q_C)}{2w(Q_C)/\Delta_C} && \text{By definition $L_C = w(Q_C)/\Delta_C$.}\\
		\nonumber && &\geq \Delta_C - \frac{\Delta_C}{2}\\
		&& &\geq \frac{\Delta_C}{2}.\label{eq:subcomponentslargevalue}
	\end{flalign}
	Combining (\ref{eq:augmentvaluesubcomponents}) and (\ref{eq:subcomponentslargevalue}) we obtain our desired property that the total augmentation value of all sub-components is, in expectation at least $p^{2L_C}\Delta_C/2$, concluding the proof of Lemma~\ref{lem:augpathaddition}.
\end{proof}

\subsection{Proof of Lemma~\ref{lem:shortpaths}}
\begin{proof}[Proof of Lemma~\ref{lem:shortpaths}]
	By Definition~\ref{def:augpathprops}, the values of all augmenting components in $U_r$ add up to at least $w(M_r) - w(O_{r-1})$, i.e.,
	\begin{equation*}
		\sum_{C \in U_r}\Delta_C \geq w(M_r) - w(O_{r-1}).
	\end{equation*}
	Separating the augmenting components based on their normalized length, we get
	\begin{flalign}
		\nonumber\sum_{C \in U_r: L_C \geq 2/\epsilon}\Delta_C + \sum_{C \in U_r: L_C < 2/\epsilon}\Delta_C &\geq w(M_r) - w(O_{r-1})\\
		\sum_{C \in U_r: L_C < 2/\epsilon}\Delta_C &\geq w(M_r) - w(O_{r-1}) - \sum_{C \in U_r: L_C \geq 2/\epsilon}\Delta_C.\label{eq:sumdeltasmall}
	\end{flalign}
	Next, we show that the total values of augmenting components with normalized length at least $2/\epsilon$ is small. Since by Definition~\ref{def:augpathprops}, $L_C = w(Q_C) / \Delta_C$, we have
	\begin{equation*}
		\sum_{C \in U_r: L_C \geq 2/\epsilon}\Delta_C = \sum_{C \in U_r: L_C \geq 2/\epsilon} \frac{w(Q_C)}{L_C}. 
	\end{equation*}
	On the other hand, since the normalized length of each component in the summation above is at least $2/\epsilon$, we have
	\begin{equation*}
		\sum_{C \in U_r: L_C \geq 2/\epsilon} \frac{w(Q_C)}{L_C} \leq \sum_{C \in U_r: L_C \geq 2/\epsilon} \frac{w(Q_C)}{2/\epsilon} \leq \frac{\epsilon}{2}\sum_{C \in U_r: L_C \geq 2/\epsilon} w(Q_C) \leq \frac{\epsilon}{2} w(M_r), 
	\end{equation*}
	which combined with the equation above implies
	\begin{equation}\label{eq:sumdeltalong}
		\sum_{C \in U_r: L_C \geq 2/\epsilon}\Delta_C \leq \frac{\epsilon}{2} w(M_r).
	\end{equation}
	Combining (\ref{eq:sumdeltalong}) and (\ref{eq:sumdeltasmall}), we get
	\begin{flalign*}
		&& \sum_{C \in U_r: L_C < 2/\epsilon}\Delta_C &\geq w(M_r) - w(O_{r-1}) - \frac{\epsilon}{2} w(M_r) && \\
		&& &= (1-\frac{\epsilon}{2})w(M_r) - w(O_{r-1})\\
		&& &\geq (1-\frac{\epsilon}{2})w(M_r) - (1-\epsilon)w(M_r) && \text{Since $w(O_{r-1}) \leq (1-\epsilon)w(M_r)$.}\\
		&& &= \frac{\epsilon}{2}w(M_r),
	\end{flalign*}
	which is the desired lower bound.
\end{proof}

\section{$(0.5 - \epsilon)$-Approximation via a Non-adaptive Algorithm}
In this section, we give a non-adaptive algorithm to obtain a $(1/2-\epsilon)$-approximation of the omniscient optimum. A non-adaptive algorithm should return a degree bounded (bounded by a constant) subset $H$ of the original graph's edge set $E$ with the goal to make $\expmatching{H}$ as close as possible to $\expmatching{E}$ where we use \expmatching{E'} to denote the expected weight of the maximum weighted matching of an edge set $E'$ given that each of its edges is realized independently with probability $p$.

To illustrate why it is not trivial to give a constant approximation, note that the realization probability $p$ could be any arbitrarily small constant. As such, to achieve a constant approximation, we need to pick a subset of edges that \emph{augment} each other well. For instance, by only picking the maximum weighted matching $M$ of the input graph $G$ as our subgraph $H$, the expected matching weight would only be $p \cdot w(M)$ which might be much smaller than the omniscient optimum, especially if $p$ is small.

Our algorithm, which is formally given as Algorithm~\ref{alg:nonadaptive}, is as follows: pick a maximum weighted matching, add it to the solution and remove its edges. We repeat this for $R = O(\frac{1}{\epsilon p^{4/\epsilon}})$ steps and prove it achieves a $(1/2 -\epsilon)$-approximation.

\begin{algorithm}
  \caption{Non-adaptive algorithm for weighted stochastic matching: $(1/2-\epsilon)$-approximation.}
  \label{alg:nonadaptive}
  \begin{algorithmic}[1]
  	\Statex \textbf{Input:} Input graph $G=(V, E)$.
  	\Statex \textbf{Parameter:} $R$.
  	\State $\Estar{} \gets E$
	\For{$r = 1, \ldots, R$}
		\State Find a maximum weighted matching \M{r} in $(V, \Estar{})$.
		\State $\Estar{} \gets \Estar - \M{r}$
	\EndFor
	\State Query the edges in $\cup_{r=1}^{R}\M{r}$ and return the maximum realized matching in it.
  \end{algorithmic}
\end{algorithm}

We refer to the iterations of the for loop in Algorithm~\ref{alg:nonadaptive} as ``rounds'' here. However, we emphasize that this is different from the adaptivity rounds of Algorithm~\ref{alg:adaptive} since here we only query once at the end of algorithm in contrast to Algorithm~\ref{alg:adaptive} that queries the picked edges at each round.

\begin{theorem}\label{thm:nonadaptive}
	For any graph $G=(V, E)$, and any arbitrarily small constant $\epsilon > 0$, Algorithm~\ref{alg:nonadaptive} returns a matching whose expected weight is at least $(1/2-\epsilon)\expmatching{E}$ for $R = O(\frac{1}{\epsilon p^{4/\epsilon}})$.
\end{theorem}

\begin{proof} The proof is similar to the proof of the adaptive algorithm with some minor differences that we cover here. The main difference that indeed causes the approximation factor to be $1/2-\epsilon$ instead of $1-\epsilon$ is that the weight of matching $M_r$ that we pick at round $r$ might be smaller than that of the omniscient optimum. However, denoting by $H_r := \cup_{i=1}^{r}M_r$ our subgraph by round $r$, we shows that until we obtain a $(1/2-\epsilon)$-approximation, the weight of the next matching that we pick is at least $1/2\expmatching{E}$. We start with the following auxiliary observation.

\begin{observation}\label{obs:gooli}
	For any partitioning of the edge set $E$ into two subsets $E_1$ and $E_2$ with $E_1 \cup E_2 = E$, we have $w(\matching{E_1}) + w(\matching{E_2}) \geq w(\matching{E})$.
\end{observation}
\begin{proof}
	 Let $\mu$ be the maximum weighted matching of $E$, and define $\mu_1 = \mu \cap E_1$ and $\mu_2 = \mu \cap E_2$. Observe that $\mu_1$ and $\mu_2$ are both valid matchings of $E_1$ and $E_2$ respectively, which implies $w(\matching{E_1}) \geq w(\mu_1)$ and $w(\matching{E_2}) \geq w(\mu_2)$. Also note that $w(\mu_1) + w(\mu_2) \geq w(\mu)$ since $E_1 \cup E_2 = E$. Combining these two observations, we get $$w(\matching{E_1}) + w(\matching{E_2}) \geq w(\mu_1) + w(\mu_2) \geq w(\mu) = w(\matching{E}),$$ completing the proof.
\end{proof}

\begin{lemma}\label{lem:nextishalf}
	For any round $r$ of Algorithm~\ref{alg:nonadaptive}, if $\expmatching{H_{r-1}} < \expmatching{E}/2$, then $w(M_r) \geq \expmatching{E}/2$.
\end{lemma}
\begin{proof}
	Let us, for any $r$, denote by $N_r$ the edges in $E \backslash H_r$. By this definition, we have $N_r \cup H_r = E$. This further implies that for any subset $E'$ of $E$, we have $(N_r \cap E') \cup (H_r \cap E') = E'$. Combining this with Observation~\ref{obs:gooli}, we get that for any realization $G_p=(V, E_p)$ of $G$, we have $w(\matching{H_r \cap E_p}) + w(\matching{N_r \cap E_p}) \geq w(\matching{E_p})$. Thus, we have
	\begin{flalign*}
		\E_{E_p} \Big[ w(\matching{H_r \cap E_p}) + w(\matching{N_r \cap E_p}) \Big] &\geq \E_{E_p} \Big[w(\matching{E_p})\Big]\\
		\E_{E_p} \Big[ w(\matching{H_r \cap E_p})\Big] + \E_{E_p}\Big[w(\matching{N_r \cap E_p}) \Big] &\geq \E_{E_p} \Big[w(\matching{E_p})\Big]\\
		\expmatching{H_r} + \expmatching{N_r} &\geq \expmatching{E}.
	\end{flalign*}
	Now, having $\expmatching{H_{r-1}} < \expmatching{E}/2$ implies that $\expmatching{N_{r-1}} \geq \expmatching{E}/2$ meaning that there is at least one matching of weight at least $\expmatching{E}/2$ in $N_{r-1}$. Since all the edges in $N_{r-1}$ can be chosen in the matching $M_r$ picked in the next round, we have $w(M_r) \geq w(\matching{N_{r-1}}) \geq \expmatching{E}/2$ as desired.
\end{proof}

Our goal is to show that if $\expmatching{H_{r-1}} < (1/2 - \epsilon)\expmatching{E}$ then $\expmatching{H_r} - \expmatching{H_{r-1}}$ is larger than $\Omega(p^{4/\epsilon}\epsilon\expmatching{E})$ which implies that after $O(\frac{1}{\epsilon p^{4/\epsilon}})$ rounds we obtain a $(1/2-\epsilon)$-approximation. To do so, similar to the adaptive case, we maintain the maximum realized weighted matching $O_r$ obtained by round $r$ of Algorithm~\ref{alg:nonadaptive} and show that if $O_{r-1}$ is smaller than $(1/2-\epsilon)\expmatching{E}$, the next matching $M_r$ augments $O_{r-1}$ by at least $\epsilon p^{4/\epsilon}\expmatching{E}$ in expectation.

Recall that the argument in proving Theorem~\ref{thm:adaptive} was based on the fact (Observation~\ref{obs:standardforadaptive}) that the matching that we pick in each round has weight at least $\opt$, and argued that querying its edges increases the weight of our maintained realized matching by a large factor in expectation. For the non-adaptive algorithm, however, the matching that we pick in each round might not be this large since we remove the previously picked edges from the graph. Instead, we know by Lemma~\ref{lem:nextishalf}, that until we obtain a $(1/2-\epsilon)$-approximation, its weight $w(M_r)$ is at least $\expmatching{E}/2$. Therefore, we can adapt\footnote{We emphasize that it is crucial that the proofs of these two lemmas do not depend at all on the adaptivity of Algorithm~\ref{alg:adaptive}.} Lemmas~\ref{lem:augpathaddition} and \ref{lem:shortpaths} to argue that we iteratively improve the weight of our maintained matching to obtain a $(1/2-\epsilon)$-approximation in only $O(\frac{1}{\epsilon p^{4/\epsilon}})$ iterations of Algorithm~\ref{alg:nonadaptive}. 

 \end{proof}

\section{Open Problems \& Future Directions}
An interesting future direction is to improve the approximation factor of our non-adaptive algorithm for weighted matchings. It was shown by Assadi et al. \cite{DBLP:conf/sigecom/AssadiKL17} that it is possible to obtain a slightly better than half approximation for the unweighted case, however, their analysis does not carry over to weighted graphs since it is highly tailored for cardinality matchings.

Another open question is how close we can get to the known lower bound of $\Omega(\log (1/\epsilon) / p)$ on the number of per-vertex queries to achieve a $(1-\epsilon)$-approximation adaptive algorithm. For the unweighted case, an algorithm with linear dependence on both $\epsilon$ and $p$ is known \cite{DBLP:conf/sigecom/AssadiKL16}. Although the techniques there were   generalized to the weighted case by Maehara and Yamaguchi \cite{maeharatakanori}, the per-vertex queries --- despite not having an exponential dependence on $\epsilon$ and $p$ --- depends on structural properties of the graph such as the number of vertices or the edge-weights which we showed are not necessary.

\section{Acknowledgements}
We thank Mahsa Derakhshan for her helpful suggestions and ideas. We are also thankful to anonymous reviewers for their helpful comments.


\bibliographystyle{plain}
\bibliography{references}

\begin{thebibliography}{10}

\bibitem{adamczyk2011improved}
Marek Adamczyk.
\newblock {Improved analysis of the greedy algorithm for stochastic matching}.
\newblock {\em Information Processing Letters}, 111(15):731--737, 2011.

\bibitem{akbarpour2014dynamic}
Mohammad Akbarpour, Shengwu Li, and Shayan~Oveis Gharan.
\newblock Dynamic matching market design.
\newblock {\em arXiv preprint arXiv:1402.3643}, 2014.

\bibitem{anderson2015dynamic}
Ross Anderson, Itai Ashlagi, David Gamarnik, and Yash Kanoria.
\newblock A dynamic model of barter exchange.
\newblock In {\em Proceedings of the twenty-sixth annual ACM-SIAM symposium on
  Discrete algorithms}, pages 1925--1933. Society for Industrial and Applied
  Mathematics, 2015.

\bibitem{anderson2015finding}
Ross Anderson, Itai Ashlagi, David Gamarnik, and Alvin~E Roth.
\newblock Finding long chains in kidney exchange using the traveling salesman
  problem.
\newblock {\em Proceedings of the National Academy of Sciences},
  112(3):663--668, 2015.

\bibitem{DBLP:conf/sigecom/AssadiKL16}
Sepehr Assadi, Sanjeev Khanna, and Yang Li.
\newblock The stochastic matching problem with (very) few queries.
\newblock In {\em Proceedings of the 2016 {ACM} Conference on Economics and
  Computation, {EC} '16, Maastricht, The Netherlands, July 24-28, 2016}, pages
  43--60, 2016.

\bibitem{DBLP:conf/sigecom/AssadiKL17}
Sepehr Assadi, Sanjeev Khanna, and Yang Li.
\newblock The stochastic matching problem: Beating half with a non-adaptive
  algorithm.
\newblock In {\em Proceedings of the 2017 {ACM} Conference on Economics and
  Computation, {EC} '17, Cambridge, MA, USA, June 26-30, 2017}, pages 99--116,
  2017.

\bibitem{awasthi2009online}
Pranjal Awasthi and Tuomas Sandholm.
\newblock Online stochastic optimization in the large: Application to kidney
  exchange.
\newblock In {\em IJCAI}, volume~9, pages 405--411, 2009.

\bibitem{bansal2012lp}
Nikhil Bansal, Anupam Gupta, Jian Li, Juli{\'a}n Mestre, Viswanath Nagarajan,
  and Atri Rudra.
\newblock {When LP is the cure for your matching woes: Improved bounds for
  stochastic matchings}.
\newblock {\em Algorithmica}, 63(4):733--762, 2012.

\bibitem{DBLP:conf/sigecom/BlumDHPSS15}
Avrim Blum, John~P. Dickerson, Nika Haghtalab, Ariel~D. Procaccia, Tuomas
  Sandholm, and Ankit Sharma.
\newblock Ignorance is almost bliss: Near-optimal stochastic matching with few
  queries.
\newblock In {\em Proceedings of the Sixteenth {ACM} Conference on Economics
  and Computation, {EC} '15, Portland, OR, USA, June 15-19, 2015}, pages
  325--342, 2015.

\bibitem{blum2013harnessing}
Avrim Blum, Anupam Gupta, Ariel Procaccia, and Ankit Sharma.
\newblock {Harnessing the power of two crossmatches}.
\newblock In {\em Proceedings of the fourteenth ACM conference on Electronic
  commerce}, pages 123--140. ACM, 2013.

\bibitem{chen2009approximating}
Ning Chen, Nicole Immorlica, Anna~R Karlin, Mohammad Mahdian, and Atri Rudra.
\newblock {Approximating matches made in heaven}.
\newblock In {\em International Colloquium on Automata, Languages, and
  Programming}, pages 266--278. Springer, 2009.

\bibitem{costello2012stochastic}
Kevin~P Costello, Prasad Tetali, and Pushkar Tripathi.
\newblock {Stochastic matching with commitment}.
\newblock In {\em International Colloquium on Automata, Languages, and
  Programming}, pages 822--833. Springer, 2012.

\bibitem{dickerson2012dynamic}
John~P Dickerson, Ariel~D Procaccia, and Tuomas Sandholm.
\newblock Dynamic matching via weighted myopia with application to kidney
  exchange.
\newblock In {\em AAAI}, volume 2012, pages 98--100, 2012.

\bibitem{dickerson2012optimizing}
John~P Dickerson, Ariel~D Procaccia, and Tuomas Sandholm.
\newblock Optimizing kidney exchange with transplant chains: Theory and
  reality.
\newblock In {\em Proceedings of the 11th International Conference on
  Autonomous Agents and Multiagent Systems-Volume 2}, pages 711--718.
  International Foundation for Autonomous Agents and Multiagent Systems, 2012.

\bibitem{dickerson2013failure}
John~P Dickerson, Ariel~D Procaccia, and Tuomas Sandholm.
\newblock Failure-aware kidney exchange.
\newblock In {\em Proceedings of the fourteenth ACM conference on Electronic
  commerce}, pages 323--340. ACM, 2013.

\bibitem{dickerson2015futurematch}
John~P Dickerson and Tuomas Sandholm.
\newblock Futurematch: Combining human value judgments and machine learning to
  match in dynamic environments.
\newblock In {\em AAAI}, pages 622--628, 2015.

\bibitem{gupta2013stochastic}
Anupam Gupta and Viswanath Nagarajan.
\newblock {A stochastic probing problem with applications}.
\newblock In {\em International Conference on Integer Programming and
  Combinatorial Optimization}, pages 205--216. Springer, 2013.

\bibitem{maeharatakanori}
Takanori Maehara and Yutaro Yamaguchi.
\newblock {Stochastic Packing Integer Programs with Few Queries}.
\newblock In {\em Proceedings of the Twenty-Ninth Annual ACM-SIAM Symposium on
  Discrete Algorithms}. SIAM, 2018.

\bibitem{manlove2012paired}
David~F Manlove and Gregg O’Malley.
\newblock Paired and altruistic kidney donation in the uk: Algorithms and
  experimentation.
\newblock In {\em International Symposium on Experimental Algorithms}, pages
  271--282. Springer, 2012.

\bibitem{unver2010dynamic}
M~Utku {\"U}nver.
\newblock Dynamic kidney exchange.
\newblock {\em The Review of Economic Studies}, 77(1):372--414, 2010.

\end{thebibliography}

\end{document}